\documentclass[aps,pra,superscriptaddress,showpacs,twocolumn]{revtex4-1}

\usepackage{amsfonts,amssymb,amsmath}
\usepackage{graphics,graphicx,epsfig}
\usepackage{amsthm}
\usepackage{hyperref}

\DeclareRobustCommand{\stirling}{\genfrac\{\}{0pt}{}}
\def\identity{\leavevmode\hbox{\small1\kern-3.8pt\normalsize1}}

\newtheorem{theorem}{Theorem}

\newtheorem{lemma}{Lemma}
\newtheorem{corollary}{Corollary} 

\newcommand{\ket}[1]{\left | #1 \right\rangle}

\newcommand{\avg}[1]{\left\langle #1 \right\rangle}
\renewcommand{\epsilon}{\varepsilon}
\newcommand{\Sec}[1]{\section{#1} }
\newcommand{\SubSec}[1]{\subsubsection{\NoCaseChange{#1}} }
\begin{document}

\title{Lieb-Robinson bounds on $n$-partite connected correlation functions}

\author{Minh~Cong~Tran}
\affiliation{Joint Center for Quantum Information and Computer Science, NIST/University of Maryland, College Park, Maryland 20742, USA}
\affiliation{Joint Quantum Institute, NIST/University of Maryland, College Park, Maryland 20742, USA}

\author{James R. Garrison}
\affiliation{Joint Center for Quantum Information and Computer Science, NIST/University of Maryland, College Park, Maryland 20742, USA}
\affiliation{Joint Quantum Institute, NIST/University of Maryland, College Park, Maryland 20742, USA}

\author{Zhe-Xuan Gong}
\affiliation{Joint Center for Quantum Information and Computer Science, NIST/University of Maryland, College Park, Maryland 20742, USA}
\affiliation{Joint Quantum Institute, NIST/University of Maryland, College Park, Maryland 20742, USA}
\affiliation{Department of Physics, Colorado School of Mines, Golden, Colorado 80401, USA}
\author{Alexey V. Gorshkov}
\affiliation{Joint Center for Quantum Information and Computer Science, NIST/University of Maryland, College Park, Maryland 20742, USA}
\affiliation{Joint Quantum Institute, NIST/University of Maryland, College Park, Maryland 20742, USA}

\begin{abstract}
Lieb and Robinson provided bounds on how fast bipartite connected correlations can arise in systems with only short-range interactions. 
We generalize Lieb-Robinson bounds on bipartite connected correlators to multipartite connected correlators.
The bounds imply that an $n$-partite connected correlator can reach unit value in constant time.  Remarkably, the bounds also allow for an $n$-partite connected correlator to reach a value that is \emph{exponentially large} with system size in constant time, a feature which stands in contrast to bipartite connected correlations.
We provide explicit examples of such systems.

\end{abstract}


\maketitle
\section{Introduction}
Nonrelativistic quantum mechanics is not explicitly causal. 
Long-range interactions in many physical systems allow spatially separated subsystems to become correlated at arbitrarily high speed~\cite{HK06,FGCG15,EWMK13}. 
They enable superior quantum applications such as fast quantum state transfer~\cite{EGMFG16}.  
However, in finite-dimensional systems with only bounded, short-range interactions, there is a maximum speed at which correlations may grow~\footnote{In infinite-dimensional systems, e.g.\ bosons, correlations may grow at arbitrarily high speed~\cite{EG09}, provided that non-relativistic quantum mechanics still applies.}.
If a bipartite system is initially in a product state, Lieb-Robinson bounds~\cite{LR72} imply that its bipartite connected correlation function $\avg{\mathcal{A}_\mathcal{X}\mathcal{A}_\mathcal{Y}} - \avg{\mathcal{A}_\mathcal{X}}\avg{\mathcal{A}_\mathcal{Y}}$ at time $t$ is upper bounded by $\propto\exp(v_\mathrm{LR}t- r)$~\cite{BHV06,NOS06}, where $r$ is the distance between the two subsystems $\mathcal{X}$ and $\mathcal{Y}$, and $v_\mathrm{LR}$ is the time-independent Lieb-Robinson velocity.
The bounds generate an effective light cone $v_\mathrm{LR}t = r$, outside which any bipartite connected correlation function is exponentially small.

The bounds of Lieb and Robinson are useful in many contexts~\cite{NS10,KGE13,NRSS09,SHKM10,SH10}.
Recent experiments have measured the precise shape of the light cone in many-body systems~\cite{CBPESFGBKK,RGLSSFMGM14}.  In one case, a faster-than-linear light cone was observed in an effective spin chain, thus indicating the presence of long-range interactions~\cite{RGLSSFMGM14}.
The bounds also have implications for quantum state preparation, as
preparation of a quantum state implies successful generation of \emph{all} of its correlations.
The Lieb-Robinson bound on bipartite connected correlations therefore imposes a lower limit for the time one needs to prepare bipartite quantum states when only bounded, short-range interactions are available.
This statement can be directly generalized for multipartite quantum states.
Lower limits for preparation time can be obtained by applying Lieb-Robinson bounds on every connected correlator between all pairs of sites in a system.
However, such two-point connected correlators do not fully characterize multipartite systems, the collective properties of which 
are better captured by multipartite connected correlators. 
For example, in pure states, multipartite correlations reveal the presence of genuine multipartite entanglement~\cite{ZZX06}.
Therefore it is natural to ask whether one may achieve better understanding of multipartite systems by examining Lieb-Robinson-like bounds on multipartite correlators. 
Such a study is timely, given the recent successful measurement of multipartite connected correlators in atomic superfluids~\cite{SKEMRCLGBS17}.

In this paper, we generalize Lieb-Robinson bounds on bipartite connected correlators to multipartite connected correlators.
We then show that there exist systems where the bounds are saturated. 
We argue that the bounds on multipartite correlations provide practical advantages over bipartite bounds.
In addition, our Lieb-Robinson bounds on multipartite connected correlators imply that exponentially large correlations can be created in fixed time, independent of a system's size.
We provide explicit examples of systems with this feature.\\

\Sec{Connected correlations}
Let us first define bipartite connected correlators.
Consider a set of $n$ sites $\Gamma$ and two distinct, non-overlapping subsets $\mathcal{X}\subset\Gamma$ and $\mathcal{Y}\subset\Gamma$.
Denote by $\mathcal{S}(\mathcal{X})$ the set of observables for which support lies entirely in $\mathcal{X}$.
The bipartite \emph{disconnected} correlator between observables $A_{\mathcal{X}}\in \mathcal{S}(\mathcal{X})$ and $A_{\mathcal{Y}}\in \mathcal{S}(\mathcal{Y})$ 
is simply the expectation value of their joint measurement outcomes at equal time, i.e.,\ $\avg{A_\mathcal{X}A_\mathcal{Y}}$.
Often in experiments only single sites are directly accessible.
Observables are then supported by single sites, i.e.,\ $|\mathcal{X}|=|\mathcal{Y}|=1$.
In the following discussions we refer to such correlators as two-point disconnected correlators.

We note that disconnected correlators contain both quantum and classical correlations.
For example in two-qubit systems, the disconnected correlator $\avg{Z_1Z_2}$ (where $Z$ is the Pauli matrix) achieves maximal value in both the fully classical state $\ket{00}$ and the maximally entangled state $\frac{1}{\sqrt2}\left(\ket{00}+\ket{11}\right)$~\cite{HHHH09}.
Their difference lies in the local expectation values $\avg{Z_1}$, $\avg{Z_2}$, which are maximal for the product state and vanish for the maximally entangled state. 
These local expectation values therefore can be said to carry classical information of the systems (in pure states).
The bipartite \emph{connected} correlator is constructed by subtracting this ``classicalness'' from the disconnected correlator:
\begin{align}
	u_2\left(A_\mathcal{X},A_\mathcal{Y}\right) \equiv \avg{A_\mathcal{X} A_\mathcal{Y}} - \avg{A_\mathcal{X}}\avg{A_\mathcal{Y}}.\label{DEF_u2}
\end{align}
In general for mixed systems, if the joint state of $\mathcal{X}\cup\mathcal{Y}$ is a product state, i.e.,~$\rho_{\mathcal{X}\cup\mathcal{Y}}=\rho_\mathcal{X}\otimes\rho_\mathcal{Y}$, its disconnected correlators $\avg{A_\mathcal{X} A_\mathcal{Y}}$ are factorizable into $\avg{A_\mathcal{X}}\avg{A_\mathcal{Y}}$ and therefore all bipartite connected correlators vanish. 
The opposite is also true~\cite{ZZX06}:
\begin{lemma}
\label{LemSep2}
A density matrix $\rho$ is a product state, i.e.,\ there exist complementary subsets $\mathcal{X},\tilde{\mathcal{X}}$ such that $\rho = \rho_\mathcal{X}\otimes\rho_{\tilde{\mathcal{X}}}$, if and only if 
\begin{align}
	u_2(A_\mathcal{X},A_{\tilde{\mathcal{X}}}) = 0,
\end{align}
for all observables $A_\mathcal{X}\in \mathcal{S}(\mathcal{X})$ and $A_{\tilde{\mathcal{X}}}\in \mathcal{S}(\tilde{\mathcal{X}})$.
\end{lemma}
In particular, a nonzero bipartite connected correlator implies bipartite entanglement in pure states.
Lemma \ref{LemSep2} is a consequence of Ref.~\cite{ZZX06}.
We also present a simple proof in Appendix \ref{APP_Sep2}.

A natural generalization of the bipartite connected correlator to multipartite systems is the Ursell function~\cite{U27,S75}. 
The $n$-partite connected correlator between $n$ observables $A_1,\dots,A_n$, which are supported by $n$ distinct subsets of sites $\mathcal{X}_1,\dots,\mathcal{X}_n$, respectively, is defined as
\begin{align}
		u_n \left(A_1,\dots, A_n\right) = \sum_{P} g\left(|P|\right)\prod_{p\in P}\avg{\prod_{j\in p }A_j},\label{EQ_npointdef}
\end{align}
where $g(x)=(-1)^{x-1}(x-1)!$ and the sum is taken over all partitions $P$ of the set $\left\{1,2,\dots,n\right\}$.
The $n$-partite connected correlators can be equivalently defined via either recursive relations or generating functions (see Appendix \ref{APP_defnpoint} for details).

Multipartite connected correlators also arise naturally in many other contexts. 
In quantum field theory, connected Green's functions are multipartite connected correlators of field operators~\cite{H94}. 
Mean field theory is an approximation in which it is assumed that all connected correlators vanish~\cite{K10}; 
in fact, mean field theory fails when there exist significant connected correlations, and one must then seek higher-order approximations.
The cumulant expansion technique is similar to mean field theory, but only multipartite connected correlators of high enough order are ignored.
Therefore, understanding when connected correlations are negligible is important for validating mean field theory and the cumulant expansion.

The relation mentioned above between connected correlators and entanglement holds for $n$-partite connected correlators as well.
It also follows from Ref.~\cite{ZZX06} that $n$-partite connected correlators vanish in product states.
In particular, for pure states, a nonzero $n$-partite connected correlator implies genuine $n$-partite entanglement~\cite{WGE16,PRSF17}:
\begin{lemma}
\label{lemma_sep}
If an $n$-partite system is in a product state, i.e.\ there exist complementary subsystems $\mathcal{X},\bar{\mathcal{X}}\subset S_n$ such that
\begin{align}
	\rho = \rho_{\mathcal{X}}\otimes \rho_{\bar{\mathcal{X}}},
\end{align}
then all $k$-body connected correlators ($2\leq k\leq n$) between some observables $A_1,\dots,A_{k_1}$, for which support lies entirely on $\mathcal{X}$, and observables $B_1,\dots,B_{k_2}$, for which support lies entirely on $\bar{\mathcal{X}}$ $(k_1,k_2\geq 1, k_1+k_2=k)$, vanish,
\begin{align}	u_k\left(A_1,\dots,A_{k_1},B_1,\dots,B_{k_2}\right) = 0.
\end{align}
\end{lemma}
\begin{corollary}
If an $n$-partite pure state $\ket{\psi}$ has a nonzero $n$-partite connected correlator, then it is genuinely $n$-partite entangled, i.e.\ there exist no subsystems $\mathcal{X}$ and $\tilde{\mathcal{X}}$ such that $\ket{\psi} = \ket{\psi_{\mathcal{X}}}\otimes\ket{\psi_{\bar{\mathcal{X}}}}$.
\end{corollary}
A direct proof of Lemma \ref{lemma_sep} is presented in Appendix \ref{APP_sep}.
The combination of Lemma 1 and Lemma 2 tells us that if the bipartite connected correlators are all zero between two regions, then all higher-order connected
correlators are guaranteed to be zero except for the scenario where all observables are supported on one region, or there exists
an observable supported on both regions.

Multipartite connected correlations also provide a practical advantage over bipartite correlations, even though the latter are sufficient to characterize a quantum system. 
Consider a three-body system for example. 
The collection of local expectation values and connected correlators,
\begin{align}
U=\bigg\{\avg{A_1},\avg{A_2}&,\avg{A_3},u_2(A_1,A_2),u_2(A_1,A_3),\nonumber\\
&u_2(A_2,A_3),u_3(A_1,A_2,A_3)\bigg\},
\end{align}
where each $A_j$ runs over a complete single site basis (e.g.\  the Pauli matrices $X,Y,Z$), defines a unique tripartite quantum state.
Another equivalent collection $\tilde U$ can be constructed from $U$ by replacing $u_3(A_1,A_2,A_3)$ with a bipartite connected correlator between one subsystem and the rest, e.g.\ $u_2(A_1,A_2A_3)$.
Although the two collections $U$ and $\tilde U$ are equivalent, $u_3(A_1,A_2,A_3)$ and $u_2(A_1,A_2A_3)$ carry different information about the system.
The three-point connected correlators $u_3(A_1,A_2,A_3)$ characterize global properties while $u_2(A_1,A_2A_3)$ only tell us about local properties across the cut between subsystem 1 and the rest.
If global properties, such as genuine three-body entanglement, are of concern, then tripartite connected correlators are superior.
To have a chance at detecting genuine tripartite entanglement using only bipartite connected correlators, one must consider all possible bipartitions of the system.
There are only 3 such partitions for a tripartite system, namely $1|23,2|13$ and $3|12$.
But for $n$-partite systems, the number of bipartitions scales exponentially with $n$.
Computing all of them would be impractical. 
Even then there is no guarantee they would detect genuine multipartite entanglement.
Consider for example the following pure state of 3 qubits,
\begin{align}
\ket{\psi}=&\sqrt{\frac{5}{24}}\ket{000} + \sqrt{\frac18}\ket{001} + \sqrt{\frac{1}{12}}\ket{010} + \sqrt\frac{1}{12}\ket{011}\nonumber\\
&+ \sqrt\frac{1}{4}\ket{100} + \sqrt\frac{1}{8}\ket{101} + \sqrt\frac{1}{12}\ket{110} + \sqrt\frac{1}{24}\ket{111}.
\end{align}
Its three-point connected correlator $u_3(Z_1,Z_2,Z_3)=\frac{1}{18}$ implies genuine tripartite entanglement in $\ket{\psi}$.
Meanwhile, non-zero bipartite connected correlators across the cuts $2|13$ and $3|12$, $u_2(Z_2,Z_1Z_3)$ and $u_2(Z_3,Z_1Z_2)$, only tell us that there is entanglement between qubits 2 and 3.
Because the bipartite connected correlator across $1|23$ $u_2(Z_1,Z_2Z_3)$ is zero, it is inconclusive whether the first qubit is entangled with the others without considering higher order correlators.

This example demonstrates why multipartite connected correlators are better candidates than bipartite counterparts in multipartite entanglement detection schemes. 
It is therefore important to understand how these multipartite correlations evolve in physical systems.
\begin{figure}[t]
    \centering
    \includegraphics[width=0.45\textwidth]{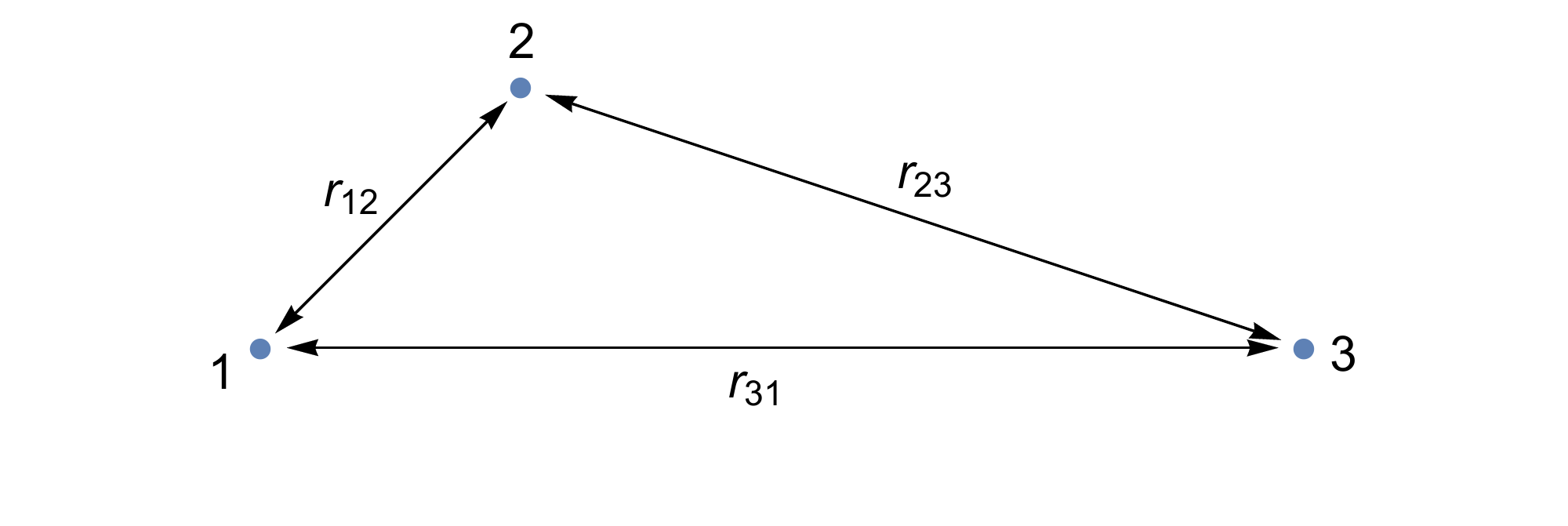}
 
    \caption{A typical three-body system. Each dot represents one site. There are three relevant length scales $r_{12},r_{23}$ and $r_{31}$. Which length scale will define the three-body Lieb-Robinson bound?}
    \label{FIG_3partite}
\end{figure}

\Sec{Multipartite Lieb-Robinson bounds}
Our main result is Lieb-Robinson-like bounds on $n$-partite connected correlators in systems evolving from fully product states under short-range interactions, e.g.\ 
\begin{align}
	H = \sum_{\avg{i,j}} J_{ij}  V_i   V_j,
\end{align}
where $ V_i$ is the spin operator of the $i^{\text{th}}$ site, $|J_{ij}|\leq 1$ is the interaction strength between the $i^{\text{th}}$ and the $j^{\text{th}}$ sites, and the sum is over all neighboring $i,j$.
But before we present the bounds, let us discuss general features we expect from such bounds. 
These bounds are of the form
\begin{align}
	u_n \leq C_n \exp(v_\mathrm{LR} t - r),\label{EQ_candidate}
\end{align}
where $C_n$ is a constant, $r$ is a relevant length scale, and $v_\mathrm{LR}$ is the same Lieb-Robinson velocity as in the bipartite bounds.
Let us now examine the scaling of $C_n$ with $n$.
If all observables have unit norm, bipartite connected correlators are upper bounded by 1 regardless of a system's size.
However, multipartite connected correlators can increase in value with the number of subsystems.
For example, in the $n$-qubit Greenberger-Horne-Zeilinger (GHZ) state,
\begin{align}
	\ket{\mathrm{GHZ}} = \frac{\ket{0}^{\otimes n}+\ket{1}^{\otimes n}}{\sqrt{2}}.
\end{align}
the $n$-point connected correlator $u_n(Z_1,\dots,Z_n) = \mathcal{O}(n^n)$ (details in Appendix \ref{APP_cal_corr}).
Therefore we expect $C_n$ to grow with $n$ as well, $C_n = \mathcal{O}(n^n)$.
Another constant we would like to understand is the critical distance $r$.
In the Lieb-Robinson bound on a bipartite connected correlator, the critical distance is simply the distance between the two involved parties.
However, in a multipartite system there are many relevant length scales which could possibly serve as the critical distance.
As an example, let us consider a three-qubit system (Fig.~\ref{FIG_3partite}).
Without loss of generality we assume $r_{12}<r_{23}<r_{31}$ where $r_{ij}$ denotes the distance between the $i^{\text{th}}$ and $j^{\text{th}}$ qubits.
We argue that a bound of the form \eqref{EQ_candidate} with $r=r_{12}$ is valid but trivial.
Intuitively an observable initially localized at the first qubit will need time to spread a distance $r_{12}$ before ``seeing'' another qubit.
Is there a stronger bound, i.e.\ inequality \eqref{EQ_candidate} with larger value for $r$? 
The largest distance $r_{31}$ would make the most sense, since at $t = r_{31}/v$, an observable initially localized at one qubit has enough time to spread to all others.
We shall show below that the critical distance for the tightest bound is neither the smallest ($r_{12}$) nor the largest distance ($r_{31}$), but actually the intermediate length scale $r_{23}$. 
This surprising result leads to unexpected consequences, including the creation of exponentially large connected correlations in unit time.

\begin{figure}[t]
    \centering
    \includegraphics[width=0.45\textwidth]{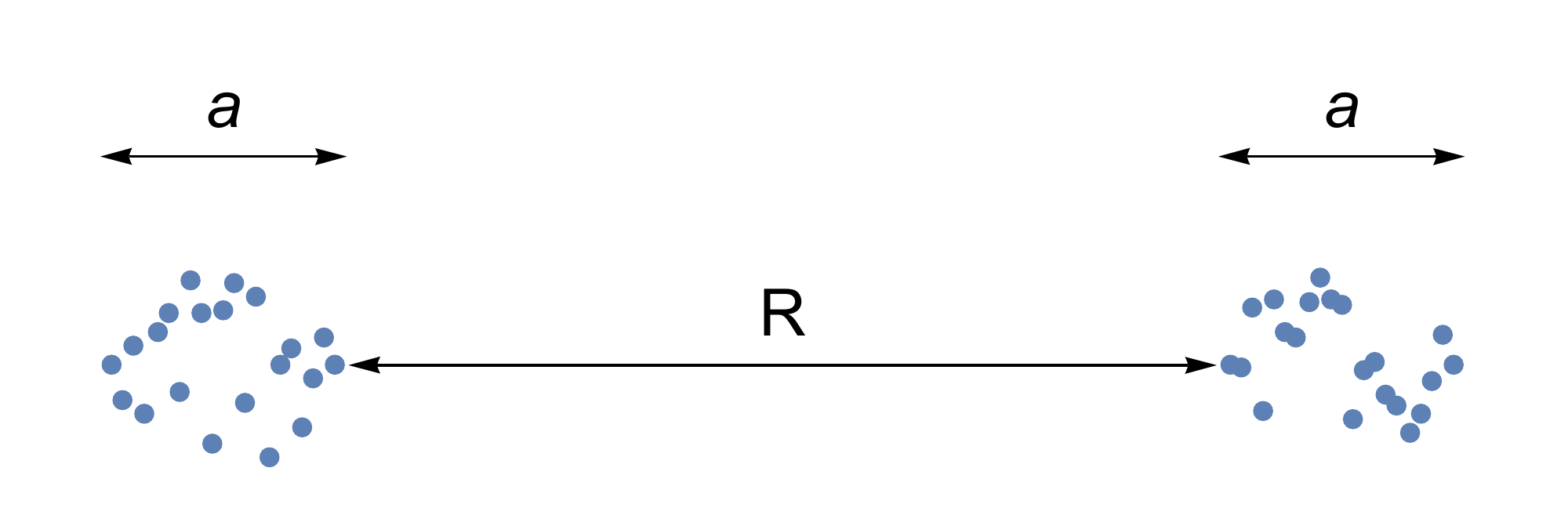}
 
    \caption{A geometry where $n$ sites (blue dots) are divided into two cliques such that the clique size $a$ is much smaller than the distance $R$ between cliques.}
    \label{FIG_GHZsaturated}
\end{figure}

\begin{theorem}
\label{TH_npointLR}
Given $n$ non-overlapping subsystems $\left\{\mathcal{X}_1,\dots,\mathcal{X}_n\right\} = S$ initialized to a fully product state $\ket{\psi_{\mathcal{X}_1}}\otimes\dots\otimes \ket{\psi_{\mathcal{X}_n}}$ and evolved under short-range interactions, the $n$-partite connected correlator between observables $A_i\in\mathcal{S}(\mathcal{X}_i)$ $(i=1,\dots,n)$ is bounded,
\begin{align}
	\left|u_n\left(A_1,\dots,A_n\right)\right| \leq C_n \exp(v_\mathrm{LR}t-R),\label{EQ_npointLR}
\end{align}
where $v_\mathrm{LR}$ is the same velocity as in the bipartite Lieb-Robinson bounds, $C_n=\frac{n^n}{4} C_2$ with $C_2$ being the constant in bipartite Lieb-Robinson bounds \footnote{See Eq.\ (6) of Ref.~\cite{RGLSSFMGM14}}, and
\begin{align}
	R=\max_{\mathcal{S}_1\subset S} d(\mathcal{S}_1,\bar{\mathcal{S}_1})
\end{align} 
is 
the largest distance between any subset $\mathcal{S}_1\subset\mathcal{S}$ and its complementary subset $\bar{\mathcal{S}_1}$. 
Here the distance $d$ between two sets of sites is the shortest distance between a site in one set and a site in the other set. 
\end{theorem}
\begin{proof}
We shall explain our proof in the simplest case of $n = 3$.
We use the following identity (given in Appendix \ref{APP_defnpoint}) to write disconnected correlators in terms of connected correlators,
\begin{align}
	\avg{A_1A_2A_3} =& u_3(A_1,A_2,A_3) + u_2(A_2,A_3)\avg{A_1}\nonumber\\
    &+u_2(A_1,A_3)\avg{A_2}+u_2(A_1,A_2)\avg{A_3}\nonumber\\
    &+\avg{A_1}\avg{A_2}\avg{A_3}. \label{EQ_th1proof3bodyidentity}
\end{align}
Notice that the last two terms on the right hand side sum up to $\avg{A_1A_2}\avg{A_3}$.
If we move this term to the left hand side, we obtain an expression of $u_3$ in terms of only bipartite connected correlators (and local expectation values),
\begin{align}
	u_3(A_1,A_2,A_3) =& u_2(A_1A_2,A_3)-u_2(A_1,A_3)\avg{A_2}\nonumber\\
    &-u_2(A_2,A_3)\avg{A_3},
\end{align}
where the local expectation values $\avg{A_2},\avg{A_3}$ are between -1 and 1. 
Therefore we may bound the three-body connected correlator using the bipartite Lieb-Robinson bound as follows,
\begin{align}
	&\left|u_3(A_1,A_2,A_3)\right|  \nonumber\\
    &\leq  \left|u_2(A_1A_2,A_3)\right|+\left|u_2(A_1,A_3)\right|
  +\left|u_2(A_2,A_3)\right|\nonumber\\
    &\leq C_2 e^{v_\mathrm{LR}t - r_{12|3}}+C_2 e^{v_\mathrm{LR}t-r_{13}}+C_2 e^{v_\mathrm{LR}t-r_{23}}\nonumber\\
    &\leq 3C_2 e^{v_\mathrm{LR}t - r_{12|3}}, \label{EQ_th1proof12vs3}
\end{align}
where $r_{12|3}=\min\left\{r_{12},r_{13}\right\}$ is the distance from the third site to the other two and $C_2$ comes from bipartite Lieb-Robinson bounds~\cite{RGLSSFMGM14}.
One may notice that at the beginning the three sites play equal roles, but somehow this symmetry is broken in Eq.~\eqref{EQ_th1proof12vs3}.
The reason is the choice to team up $\avg{A_1}\avg{A_2}\avg{A_3}$ and $u_2(A_1,A_2)\avg{A_3}$ after Eq.~\eqref{EQ_th1proof3bodyidentity}.
Instead, we may replace the latter with either $u_2(A_2,A_3)\avg{A_1}$ or $u_2(A_1,A_3)\avg{A_2}$ to obtain two different bounds in the form of Eq.~\eqref{EQ_th1proof12vs3}, with either $r_{23|1}$ or $r_{13|2}$ in place of $r_{12|3}$.
The tightest bound corresponds to the smallest distance among $r_{23|1},r_{13|2},r_{12|3}$, and hence the theorem follows.
Proof for general $n$ follows the exact same line and is presented in full in Appendix \ref{APP_maintheoremproof}. 
\end{proof}

\begin{figure}[t]
    \centering
    \includegraphics[width=0.5\textwidth]{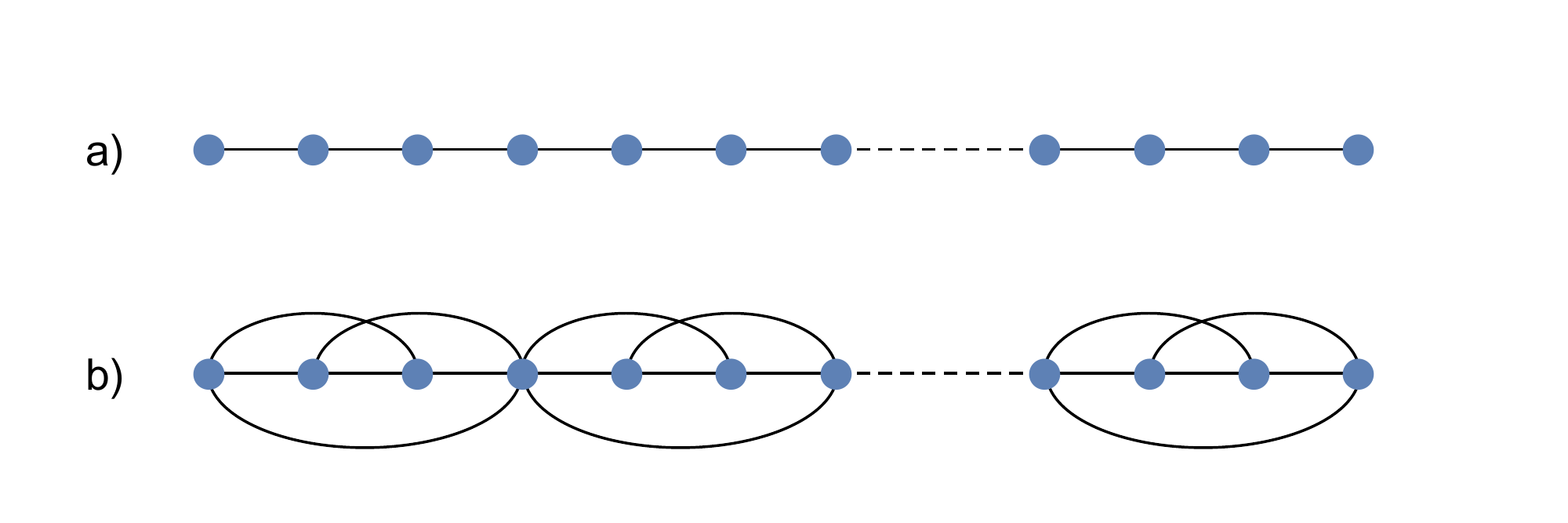}
    \caption{$n$-qubit cluster states represented by one-dimensional graphs of $n$ vertices. (a) Only consecutive vertices are connected by edges of length 1. (b) Some edges are longer than 1 but interactions are still local.}
    \label{cluster}
\end{figure}

Since the proof is inductive on the number of sites $n$, the multipartite Lieb-Robinson bounds are in general weaker than bipartite Lieb-Robinson bounds.
Violation of our bound for a multipartite connected correlator implies violation of at least one bipartite bound. 
Nevertheless, the multipartite Lieb-Robinson bounds in Theorem \ref{TH_npointLR} can be saturated.
For example, consider a geometry of $n$ sites where they are divided into two non-empty cliques, each of spatial size $a$.
The two cliques are separated by a large distance $R\gg a$ (Fig.~\ref{FIG_GHZsaturated}).
Lieb-Robinson bounds of $n$-partite connected correlators for this geometry are
saturated by preparing the GHZ state of $n$ qubits, which can be done in time $t\approx R/v_\mathrm{LR}$.

Whether our $n$-partite Lieb-Robinson bounds are tight for all geometries is still an open question.
The geometry of Fig.~\ref{FIG_GHZsaturated} resembles a bipartite system, where each clique plays the role of one party.
There are geometries which are very different from bipartite systems and, as a consequence, they offer some unique and interesting implications.
For example, as mentioned before, the critical distance in the multipartite Lieb-Robinson bound is neither the largest nor the smallest distance. 
In the asymptotic limit of large $n$, these quantities can be very different.
We shall now examine such examples.\\

\Sec{Fast generation of multipartite correlation}
In a bipartite system, the time needed to create bipartite correlators of order $\mathcal{O}(1)$ increases proportionally to the distance between the two subsystems.
It is natural to expect the time needed to create $n$-point correlators of order $\mathcal{O}(1)$ in an $n$-partite system to increase with the spatial size of the system.
However, Theorem 1 suggests that it may not necessarily be the case.
For example, consider an equally spaced one-dimensional chain of $n$ qubits [see Fig.\ \ref{cluster}].
If the distance between two consecutive qubits is fixed, the spatial length of the chain increases as $\mathcal{O}(n)$.
Therefore $2$-point connected correlators between the end qubits can only be created after $\mathcal{O}(n)$ time.
Meanwhile, Theorem 1 suggests that $n$-point connected correlators of order $\mathcal{O}(1)$ between all $n$ qubits might be created in $\mathcal{O}(1)$ time using only nearest neighbor interactions, enabling almost instant $n$-partite genuine entanglement.
As we show below, systems with such a feature do exist.

One example is the one-dimensional cluster state.
Cluster states (also called graph states) are multipartite entangled states~\cite{Hein04} useful for one-way quantum computation~\cite{RBB03,N06}.
They have a simple visual representation using associated graphs.
For a graph $G=(V,E)$, the corresponding cluster state can be constructed as follows: 
(i) Associate each vertex in $V$ with one qubit initialized in state $\ket{+}=\frac{\ket{0}+\ket{1}}{\sqrt{2}}$;
and (ii) Apply a controlled-$Z$ gate to every pair of qubits connected by an edge in $E$.
A controlled-$Z$ gate on two qubits $i$ and $j$ can be implemented by evolving the system for a time $\frac{\pi}{4}$ under the Hamiltonian,
\begin{align}
	H^{(i,j)}_{cZ} = \mathbb{I} + Z_i+Z_j-Z_iZ_j, \label{EQ_H_cZ}
\end{align}
where $Z$ is the diagonal Pauli matrix.
Some cluster states, e.g.\  Fig.~\ref{cluster}, only require application of finite-range controlled-$Z$s. 
Meanwhile, the generating Hamiltonians \eqref{EQ_H_cZ} commute with each other and therefore they can be applied simultaneously. 
Therefore such cluster states as well as their correlations can be created in constant time $\mathcal{O}(1)$.
In particular, within an $n$-independent time $\frac{\pi}{4}$ we can create $|u_n(Y_1,X_2,X_3\dots ,X_{n-1},Y_n)| = 1$ in a cluster state with only nearest neighbor interactions (Fig.~\ref{cluster}(a).
This example shows that $n$-point connected correlators of order $\mathcal{O}(1)$ can be created in unit time $\mathcal{O}(1)$, independent of a system's size.
Yet, we can do better, i.e.\ we can create \emph{exponentially large} $n$-point connected correlators in unit time.
For example, in the cluster state of Fig.~\ref{cluster}b, we allow each site to interact within a larger neighborhood. 
It still takes $\frac{3\pi}{4}=\mathcal{O}(1)$ unit time to prepare the state, while direct calculation shows that one of its correlators grows exponentially as $2^{\frac{n-1}{3}}$ (Appendix \ref{APP_cal_corr}).

 \begin{figure}[t]
    \centering
    \includegraphics[width=0.45\textwidth]{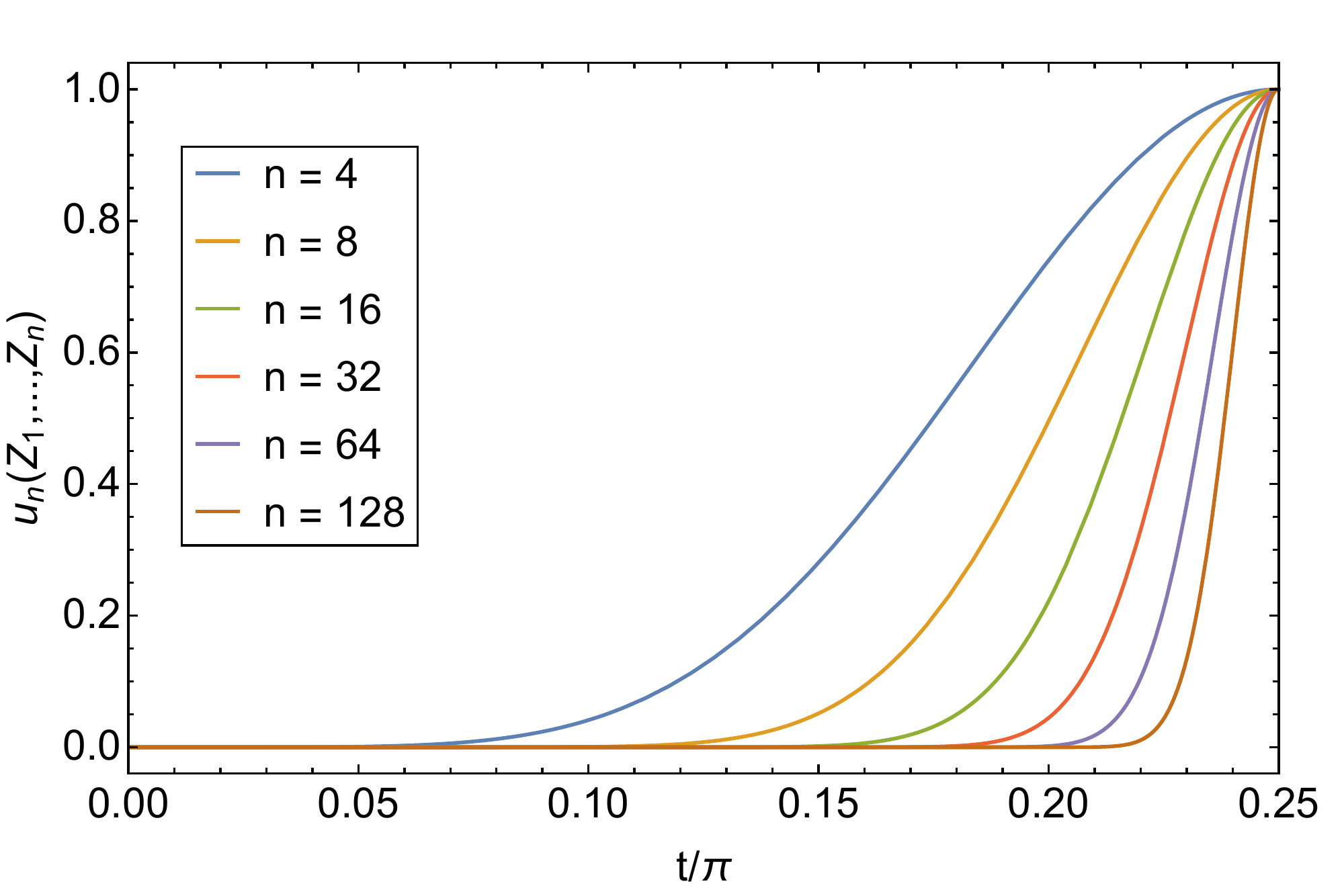}
    \caption{Time evolution of the $n$-point connected correlator $u_2(Z_1,Z_2,\dots, Z_n) = \left[\sin^2(2t)\right]^{n-1}$ of the state in Fig.~\ref{cluster}a for different $n$. Here we plot the time-dependent correlator for a few values of $n$. In the limit of large $n$ the correlator remains zero for most of the time before briefly jumping to 1 at $t = \frac{\pi}{4}$.}
    \label{fig_corrtime}
\end{figure}
In the above examples we have discussed how much time it takes to grow connected correlations from fully uncorrelated states.
A relevant question is whether it can be expedited by some initial correlations~\cite{K15}.
To answer this question, we look at the time dependence of connected correlators in an $n$-qubit system initialized to $\ket{00\dots 0}$ and evolved under the Hamiltonian $\sum_{\avg{i,j}} X_iX_j$.
If this system has the geometry of Fig.~\ref{cluster}a, we find the $n$-point connected correlator $u_n(Z_1,\dots,Z_n) = \left[\sin^2(2t)\right]^{n-1}$ (Appendix \ref{APP_cal_corr}).
We plot this function for a few values of $n$ in Fig.~\ref{fig_corrtime}.
For large $n$ the correlator remains negligible for most of the time and rapidly grows to 1 near $t = \frac{\pi}{4}$.
In other words, the connected correlator only needs a very small time $\delta t\ll 1$ to grow from almost zero to a significant value.
It gives evidence that creation of multipartite states can be expedited by small initial correlations.
We remark that while the exact correlator $u_n(Z_1,\dots,Z_n)$ is negligible at any time before $\frac{\pi}{4}$, there may exist other sets of observables for which $n$-point connected correlators are non-negligible.

\Sec{Outlook}
Although the relation between genuine multipartite entanglement and multipartite connected correlations is simple for pure states, whether one can infer any information about genuine multipartite entanglement of a mixed state from its multipartite correlations is still an open question.

In our model, only short-range interactions between two sites are present. 
An immediate question is how the Lieb-Robinson bounds generalize to other models with long-range interactions or interaction terms which involve more than two sites.

Current techniques to measure multipartite connected correlators require statistics of all measurement outcomes that factor into Eq.~\eqref{EQ_npointdef}.
Connected correlators up to tenth order have been measured using this approach~\cite{SKEMRCLGBS17}. 
However, such a method is infeasible for connected correlators of very high order,
as the number of disconnected correlators that must be measured grows exponentially with $n$. 
It is therefore an open question whether there exist experimentally accessible observables, e.g.\ magnetization~\cite{GBSWBR16}, which manifest multipartite connected correlators directly.

\acknowledgments
We thank Z. Eldredge, D. Ferguson, M. Foss-Feig, and J. Schmiedmayer for helpful discussions.
This project is supported by the NSF-funded Physics Frontier Center at the JQI, the NIST NRC Research Postdoctoral Associateship Award,
NSF QIS, AFOSR, ARO MURI, ARL CDQI, and ARO.


\appendix

\Sec{Proof of Lemma \ref{LemSep2}}
\label{APP_Sep2}

In this section we provide a proof of Lemma \ref{LemSep2}. 
One direction of the lemma is straightforward.
If the joint state is a product, i.e.\ $\rho = \rho_\mathcal{X}\otimes \rho_{\tilde{\mathcal{X}}}$, then all bipartite disconnected correlators between $A_\mathcal{X}\in \mathcal{S}(\mathcal{X})$ and $A_{\tilde{\mathcal{X}}}\in \mathcal{S}({\tilde{\mathcal{X}}})$ are factorizable, $\avg{A_\mathcal{X}A_{\tilde{\mathcal{X}}}}=\avg{A_\mathcal{X}}\avg{A_{\tilde{\mathcal{X}}}}$.
Therefore all bipartite connected correlators vanish.
To prove the opposite direction, that is vanishing of all bipartite connected correlators implies $\rho$ is a product state, let $\left\{\Gamma^\mathcal{X}_\mu\right\}$ denote a complete normalized basis for density matrices of $\mathcal{X}$, and likewise for $\left\{\Gamma^{\tilde{\mathcal{X}}}_\mu\right\}$.
Any joint state of $\mathcal{X}$ and ${\tilde{\mathcal{X}}}$ may be written as
\begin{align}
	\rho = \frac{1}{N}\bigg(\mathbb{I}_{\mathcal{X}\cup{\tilde{\mathcal{X}}}}&+\sum_\mu \avg{\Gamma^\mathcal{X}_\mu}\Gamma^\mathcal{X}_\mu\otimes \mathbb{I}_{{\tilde{\mathcal{X}}}}+\sum_\nu \avg{\Gamma^{\tilde{\mathcal{X}}}_\nu} \mathbb{I}_{\mathcal{X}}\otimes \Gamma^{\tilde{\mathcal{X}}}_\nu\nonumber\\
    &+\sum_{\mu,\nu}\avg{\Gamma^{\mathcal{X}}_\mu\otimes\Gamma^{\tilde{\mathcal{X}}}_\nu}\Gamma^{\mathcal{X}}_\mu\otimes\Gamma^{\tilde{\mathcal{X}}}_\nu\bigg),
\end{align}
where $N = |\mathcal{H}_{\mathcal{X}}\otimes \mathcal{H}_{{\tilde{\mathcal{X}}}}|$ is the dimension of the joint Hilbert space.
Since all bipartite connected correlators vanish, 
\begin{align}
	\avg{\Gamma^{\mathcal{X}}_\mu\otimes\Gamma^{\tilde{\mathcal{X}}}_\nu} = \avg{\Gamma^{\mathcal{X}}_{\mu}}\avg{\Gamma^{\tilde{\mathcal{X}}}_\nu}
\end{align}
for all $\mu,\nu$.
Therefore $\rho$ is also factorizable,
\begin{align}
\rho = \frac{1}{N} \bigg(\mathbb{I}_\mathcal{X}+\sum_\mu \avg{\Gamma^\mathcal{X}_\mu}\Gamma^\mathcal{X}_\mu\bigg)\otimes \bigg(\mathbb{I}_{\tilde{\mathcal{X}}}+\sum_\nu \avg{\Gamma^{\tilde{\mathcal{X}}}_\nu}\Gamma^{\tilde{\mathcal{X}}}_\nu\bigg).
\end{align}
Thus the lemma follows.

\Sec{Equivalent definitions of multipartite connected correlator}
\label{APP_defnpoint}
In this section we present some definitions of the multipartite connected correlation function which are equivalent to Eq.~\eqref{EQ_npointdef}.
The multipartite connected correlator can also be generated by \citep{S75}:
\begin{align}
	u_n(A_1,\dots,A_n) = \left[\frac{\partial^n}{\partial\lambda_1\dots\partial\lambda_n}\ln\avg{e^{\sum_{i=1}^n \lambda_i A_i}}\right]_{\vec \lambda=0}\label{DEF_npoint2},
\end{align}
where the partial derivative is evaluated at $\vec\lambda = (\lambda_1,\dots,\lambda_n)=0$.
This generating form will be used in Appendix \ref{APP_cal_corr} to evaluate multipartite connected correlators of the GHZ state.
An equivalent way to define multipartite connected correlators is via lower-order correlators,
\begin{align}
	u_n \left(A_1,\dots, A_n\right) =\avg{A_1\dots A_n} -\sum_{P}^{}{}^{'} \prod_{p\in P}u_{|p|}\left(\tilde A_p\right),\label{DEF_npoint3}
\end{align}
where the sum $\sum_{P}^{'}$ is taken over all partitions of $\left\{\mathcal{X}_1,\dots,\mathcal{X}_n\right\}$ except for the trivial partition $P = \left\{\mathcal{X}_1,\dots,\mathcal{X}_n\right\}$, and $\tilde A_p = \left\{A_i:i\in p\right\}$ denotes the set of all observables with indices in set $p$.
We shall find this definition useful for the inductive proof of Theorem \ref{TH_npointLR} and in Appendix \ref{APP_cal_corr}.

\Sec{Proof of Lemma \ref{lemma_sep}}
\label{APP_sep}
In this section we prove the connection between factorizability and vanishing connected correlators in Lemma \ref{lemma_sep}.
We shall prove this lemma inductively using generating functions of multipartite connected correlators \eqref{DEF_npoint2},
\begin{align}
	&\ln \avg{\exp\left({\sum_{i=1}^{k_1}\lambda_{i}A_{i}+\sum_{j=1}^{k_2}\lambda'_{j}B_{j}}\right)} \nonumber\\
    &=\ln \avg{\exp\left(\sum_{i=1}^{k_1}\lambda_{i}A_{i}\right)}+\ln\avg{\exp\left(\sum_{j=1}^{k_2}\lambda'_{j}B_{j}\right)}.
\end{align}
The first term on the right hand side is independent of any $\lambda'_j$.
Therefore, partial derivatives with respect to $\lambda'_j$s will make the first term vanish.
Similarly, the second term will also vanish after partial derivatives with respect to $\lambda_i$s.
Therefore multipartite connected correlators, which are $n^{\text{th}}$ order partial derivatives of the left hand side with respect to both $\lambda_i$s and $\lambda'_j$s, will vanish.
The lemma follows.

\Sec{Proof of Theorem 1}
\label{APP_maintheoremproof}
In this section we prove Theorem 1 by induction on $n$.
When $n = 2$, the inequalities reduce to bipartite Lieb-Robinson bounds.
Assuming that it holds for any $2\leq n \leq k-1$, we shall prove that it holds for $n = k$.
We start with the recursive definition of connected correlators (Appendix \ref{APP_defnpoint}):
\begin{align}
	\avg{A_1\dots A_k} = \sum_{P\in \mathcal{P}(S)} \prod_{p\in P} u_{|p|}\bigg(\tilde A_p \bigg), \label{DEF_Un}
\end{align}
where $\mathcal{P}(S)$ denotes the set of all partitions of $S={1,\dots,k}$, and $\tilde A_p = \left\{A_i:i\in p\right\}$ denotes the set of all observables with indices in set $p$.
Consider one particular bipartition of $S$, e.g.\  $S = S_1\cup S_2$ such that $S_1\cap S_2 = \varnothing$.
The partitions of $S$ can then be divided into two types. 
Partitions of the first type have elements that lie entirely on either $S_1$ or $S_2$.
They therefore belong to the set $\mathcal{P}(S_1)\oplus\mathcal{P}(S_2)$.
The sum over these partitions in Eq.~\eqref{DEF_Un} can then be factored into a product of two sums over $\mathcal{P}(S_1)$ and $\mathcal{P}(S_2)$,
\begin{align}
	&\bigg[\sum_{P_1\in \mathcal{P}(S_1)} \prod_{p_1\in P_1} u_{|p_1|}\bigg(\tilde A_{p_1}\bigg)\bigg]\bigg[\sum_{P_2\in \mathcal{P}(S_2)} \prod_{p_2\in P_2} u_{|p_2|}\bigg(\tilde A_{p_2}\bigg)\bigg]\nonumber\\
    &=\avg{\prod_{i\in S_1} A_{i}}\avg{\prod_{i\in S_2} A_{i}},
\end{align}
where we have used the definition \eqref{DEF_Un} for the sets $S_1$ and $S_2$.
The terms in Eq.~\eqref{DEF_Un} we have not yet summed over are partitions in which some elements overlap with both $S_1$ and $S_2$, namely $\mathcal{P}(S)\setminus\mathcal{P}(S_1)\oplus\mathcal{P}(S_2)\equiv \mathcal{P}_{12}$.
We can then rewrite Eq.~\eqref{DEF_Un} as

\begin{align}
	\avg{A_1\dots A_k} =&u_k(A_1,\dots ,A_k)+\avg{\prod_{i\in S_1} A_{i}}\avg{\prod_{i\in S_2} A_{i}}\nonumber\\
    &+\sum_{P_3\in \mathcal{P}_{12}} \prod_{p_3\in P_3} u_{|p_3|}\bigg(\tilde A_{p_3}\bigg). \label{EQN_8}
\end{align}
Rearranging Eq.~\eqref{EQN_8} in terms of bipartite connected correlators, we have 
\begin{align}
	u_k(A_1,\dots ,A_k) &= u_2\left(\prod_{i\in S_1}A_i,\prod_{i\in S_2}A_i\right)\nonumber\\
    &-\sum_{P_3\in \mathcal{P}_{12}} \prod_{p_3\in P_3} u_{|p_3|}\bigg(\left\{A_{i\in p_3}\right\}\bigg).
\end{align}
Therefore,
\begin{align}
	\left|u_k(A_1\dots A_k) \right|&\leq \left|u_2\left(\prod_{i\in S_1}A_i,\prod_{i\in S_2}A_i\right)\right|\nonumber\\
    &+\sum_{P_3\in \mathcal{P}_{12}} \prod_{p_3\in P_3}\left| u_{|p_3|}\bigg(\left\{A_{i\in p_3}\right\}\bigg)\right|.\label{EQN_Bound}
\end{align}
The first term is bounded by $\propto \exp(v t - d(S_1,S_2))$, where the distance between subsystems $S_1$ and $S_2$, i.e.\ $d(S_1,S_2)$, is defined as the smallest separation distance between a site in $S_1$ and a site in $S_2$.
To bound the second term, we first realize that the connected correlators here are between at most $k-1$ points, and therefore our induction hypothesis applies.
For each connected correlator $u$, there can be two possibilities. 
It can involve subsystems supported by both $S_1$ and $S_2$, or supported by either $S_1$ or $S_2$ alone.
If we sum over those of the second type, we again get expectation values which are bounded by 1.
For the connected correlator $u$ that involves qubits in both $S_1$ and $S_2$, by the induction hypothesis it is bounded by $\exp(vt-r)$, where $r$ is the largest distance between any bipartitions of the subsystems. 
By dividing those subsystems into those in $S_1$ and those in $S_2$, the distance $r$ has to be at least the one between $S_1$ and $S_2$, i.e.\ $r \geq d(S_1,S_2)$.
Therefore the second term in Eq.~\eqref{EQN_Bound} is also bounded  by $\exp(v t-d(S_1,S_2))$.
In the end, we get
\begin{align}
	\left|u_k(A_1,\dots ,A_k)\right| \leq C_k \exp\left[v_\mathrm{LR} t - d(S_1,S_2)\right] \label{Multi-LR}
\end{align}
for some constant $C_k$ to be determined.
For each choice of bipartition $\left\{S_1,S_2\right\}$, we get one such inequality.
The tightest bound is obtained from the bipartition with the largest distance $d$, i.e.
\begin{align}
	\left|u_k(A_1,\dots ,A_k)\right| \leq C_k \exp\left[v_\mathrm{LR} t - R\right] 
\end{align}
with $R = \max_{S_1} d(S_1,S_2)$.
Thus the hypothesis is true for $n=k$, and by induction it holds for any $n$.

We now prove the second part of the theorem, i.e.\ $C_n\leq n^n \frac{C_2}{4}$.
Clearly it holds for $n = 2$.
We prove that if the statement holds up to $n = k-1$, it must also hold for $n = k$.
Recall that a $k$-point connected correlator is bounded by Eq.~\eqref{EQN_Bound}.
The first term of Eq.~\eqref{EQN_Bound} is bounded by 1.
We need to find a bound for the sum.
Note that at the critical time $t = R/v$, the only non-negligible contributing terms are those involving $S'_1\subset S_1$ and $S'_2 \subset S_2$ such that the distance between $S'_1$ and $S'_2$ is exactly $R$ (by construction the distance is at least $R$).

Let $S^{(0)}_1\subset S_1$ and $S^{(0)}_2\subset S_2$ be such that the distance between any $s_1\in S^{(0)}_1$ and $s_2\in S^{(0)}_1$ is always $R$.
The point is that only connected correlators that involve such $s_1$ and $s_2$ will contribute to the sum.
We now count the contribution from such correlators.
If we take $k_1$ subsystems from $S^{(0)}_1$, $k_2$ subsystems from $S^{(0)}_2$ and $k_3$ subsystems from $S^{(0)}_3=S \setminus S_1^{(0)}\cup S_2^{(0)}$, their contribution is $\mathcal{O}\left((k_1+k_2+k_3)^{k_1+k_2+k_3}\right)$.
Note that summing over connected correlators of leftover subsystems, we get their disconnected correlator, which is bounded by 1.
Note also that by counting this way, some terms will appear more than once, so we get a loose bound.
Denoting by $m_1,m_2,m_3$ the size of $S^{(0)}_1,S^{(0)}_2$ and $S^{(0)}_3$, we can bound the constant $C_k$ by summing over all possible choices of $k_1 + k_2 + k_3\leq k-1$,
\begin{widetext}
\begin{align}
	C_k&\leq \frac{C_2}{4} \sum_{k_1=1}^{m_1}\sum_{k_2=1}^{m_2}\sum_{k_3=0}^{m_3} \binom{m_1}{k_1}\binom{m_2}{k_2}\binom{m_3}{k_3}(k_1+k_2+k_3)^{k_1+k_2+k_3}\\
	&\leq\frac{C_2}{4}  \sum_{k_1=1}^{m_1}\sum_{k_2=1}^{m_2}\sum_{k_3=0}^{m_3} \binom{m_1}{k_1}\binom{m_2}{k_2}\binom{m_3}{k_3}(k-1)^{k_1+k_2+k_3}\\
	&=\frac{C_2}{4}(k^{m_1}-1)(k^{m_2}-1)k^{m_3}<\frac{C_2}{4} k^{m_1+m_2+m_3} =k^k \frac{C_2}{4}.
\end{align}
Thus $C_k\leq k^k \frac{C_2}{4}$ holds for $n = k$, and by induction it holds for any $n$.
\end{widetext}

\Sec{Calculation of connected correlators}
\label{APP_cal_corr}
In this section we show how connected correlators are calculated for the GHZ states, the cluster states and the product state evolved under the $XX$ Hamiltonian.
\SubSec{The GHZ state}
\label{APP_GHZ_corr}
The generating function of $u_n(Z_1,\dots,Z_n)$ evaluated for the GHZ state of $n$ qubits is
\begin{align}
	g_n&\equiv  \ln\avg{\exp\left\{\sum_{i=1}^n \lambda_i Z_i\right\}}_{\text{GHZ}}\\
    &=\ln\left[\frac{1}{2}\exp\left(\sum_{i=1}^n\lambda_i\right)+\frac{1}{2}\exp\left(-\sum_{i=1}^n\lambda_i\right)\right]\\
    &=\ln\left[\cosh\left(\sum_{i=1}^n\lambda_i\right)\right].
\end{align}
Let $\lambda\equiv \sum_{i=1}^n\lambda_i$. Then
\begin{align}
	\frac{\partial g_n}{\partial \lambda_i} = \frac{\partial g_n}{\partial \lambda}\frac{\partial \lambda}{\partial \lambda_i} = \frac{\partial g_n}{\partial \lambda}
\end{align}
for all $i$.
Therefore the multipartite connected correlator above is given by
\begin{align}
	u_n(Z_1,\dots,Z_n) = \frac{\partial^ng_n}{\partial\lambda^n}\Bigg|_{\lambda=0}=\Bigg[\frac{\partial^n}{\partial\lambda^n}\ln(\cosh\lambda)\Bigg]_{\lambda=0}.
\end{align}
Note that this connected correlator has the same parity as $n$.
Therefore for odd $n$, it vanishes. 
For even $n$, the correlator is given by
\begin{align}
	u_n = \frac{2^n(2^n-1)B_n}{n},
\end{align}
where $B_n$ is the $n^{\text{th}}$ Bernoulli number.
In the large $n$ limit, the Bernoulli number is approximated by
\begin{align}
	|B_n|\approx 4\sqrt{\frac{\pi n}{2}}\left(\frac{n}{2\pi e}\right)^{n}.
\end{align}
Therefore the $n$-point connected correlator of the GHZ state grows as $u_n\propto n^{-1/2}(\frac2{\pi e})^n n^n= \mathcal{O}(n^n)$.

\SubSec{The cluster states}
\label{APP_clustercorr}
For each vertex $i$ in a cluster state's graph, we can associate an operator $X_i\prod_{j\in\mathcal{N}(i)} Z_j$, where $\mathcal{N}(i)$ denotes the set of vertices adjacent to $i$.
These operators generate a stabilizer group of which the cluster state is a simultaneous eigenstate.
Operators outside of this group have no disconnected correlations.
Using the stabilizer group, we can count the number of contributing disconnected correlators in the definition of connected correlators \eqref{EQ_npointdef}.
For example, for the observables $Y_1,X_2,X_3,\dots,X_{n-1},Y_n$ in the cluster state in Fig.~\ref{cluster}(a), all low-order disconnected correlators vanish.
Therefore,
\begin{align}
	&u_n(Y_1,X_2,X_3,\dots ,X_{n-1},Y_n) \nonumber\\
    =& \avg{Y_1X_2X_3\dots X_{n-1}Y_n} = 1. 
\end{align}
Similarly, by direct counting we find the $n$-point connected correlator of the Fig.~\ref{cluster}(b) cluster state $u_n(\left\{T_j:j=1,\dots,n\right\}) = 2^{\frac{n-1}{3}}$, where $T_j = X_j$ for all $1< j<n$ such that $j\equiv 1$ (mod 3), and $T_j=Y_j$ otherwise.

\SubSec{The product state evolved under the $XX$ Hamiltonian}
The time evolution shown in Fig.~\ref{fig_corrtime} can be verified as follows. 
The time-dependent state of $n$ qubits evolving from $\ket{00\dots0}$ under $H = \sum_{\avg{i,j}}X_iX_j$ can be written in the form of a matrix product state,
\begin{equation}
    \ket{\psi(t)} = \sum_{i_1,\ldots,i_n \in \{0,1\}} c_{i_1 i_2 \ldots i_n} (t) \ket{i_1 i_2 \ldots i_n},
\end{equation}
the coefficients of which are given by
\begin{equation}
    c_{i_1 i_2 \ldots i_n} (t) = L_{i_1} A_{i_2}(t) A_{i_3}(t) \cdots A_{i_{n-1}}(t) R_{i_n}(t),
\end{equation}
where
\begin{align}
    L_0 &= \begin{pmatrix} 1 & 0 \end{pmatrix}, \\
    L_1 &= \begin{pmatrix} 0 & 1 \end{pmatrix}, \\
    A_0(t) &= \begin{pmatrix} \cos t & 0 \\ 0 & -i \sin t \end{pmatrix}, \\
    A_1(t) &= \begin{pmatrix} 0 & \cos t \\ -i \sin t & 0 \end{pmatrix}, \\
    R_0(t) &= \begin{pmatrix} \cos t \\ 0 \end{pmatrix}, \\
    R_1(t) &= \begin{pmatrix} 0 \\ -i \sin t \end{pmatrix}.
\end{align}
Note that this matrix product state is in left canonical form (i.e.\ $\sum_i L^\dagger_i L_i = \sum_i A^\dagger_i A_i = I$) and it is normalized ($\sum_i R_i^\dagger R_i = 1$).
Our goal is to first determine all disconnected correlators of the form $\avg{O_1 O_2 \cdots O_n}$ where $O_i$ is either $I$ or $Z$.  Because all such operators are diagonal on each site, we can write the expectation value itself as a matrix product.  In the end, we find that the disconnected correlator picks up a factor of $\cos(2t)$ for each ``boundary'' between a $Z$ operator and an $I$ operator.  For instance, on a 5-qubit system, the expectation value $\avg{Z_2 Z_3 Z_5} = \avg{IZZIZ} = \left[\cos(2t)\right]^3$, as there are 3 relevant boundaries: between qubits 1--2, 3--4, and 4--5.

From this, it is already obvious that our connected correlator $u_n\left(Z_1, \ldots, Z_n\right)$ will be some polynomial of the variable $\cos(2t)$.  Given some partition $\mathcal{P}$, we would like to determine the power to which $\cos(2t)$ is raised.  Let us, for sake of example, denote our partition by letters of the alphabet.  
On 5 qubits, ABBCA corresponds to the product of disconnected correlators $\avg{Z_1Z_5}\avg{Z_2Z_3}\avg{Z_4} = \avg{ZIIIZ}\avg{IZZII}\avg{IIIZI} = \left[\cos(2t) \right]^6$.  
In general, the product of disconnected correlators will be $\left[ \cos(2t) \right]^{2v}$ where $v$ is the number of bonds that border two distinct subsets of the partition.  (In the case of the example ABBCA, this includes each bond except the one between sites 2--3, which are both in the same subset, B.)

Now we would like to count the number of partitions which contribute to the term with power $2v$.  Because the coefficient in the connected correlator depends on the number of subsets in the partition $|\mathcal{P}|$, we must consider separately partitions with different numbers of subsets.  Given $n$ qubits, there are $n-1$ bonds between qubits.  Thus there are $\binom{n-1}{v}$ different ways to choose $v$ bonds which connect different subsets of the partition.  Given these $v$ bonds, there are $\stirling{v}{a}$ different ways to construct partitions with $(a+1)$ total subsets.  (Here, $\stirling{v}{a}$ denotes a Stirling number of the second kind.)  Thus, the number of partitions on $n$ qubits with $v$ bonds that border two distinct subsets and with $(a+1)$ total subsets is $\binom{n-1}{v} \stirling{v}{a}$.  Note that $\sum_{v=0}^{n-1} \binom{n-1}{v} \sum_{a=0}^v \stirling{v}{a}$ is equal to the $n^{\text{th}}$ Bell number $\mathcal{B}_n$, so we have indeed accounted for all possible partitions.

As mentioned above, given a partition, two factors of $\cos(2t)$ are picked up for each bond that borders two distinct subsets.  In general, we can compute the expectation value of the connected correlator from Eq.~\eqref{EQ_npointdef} as follows:
\begin{widetext}
\begin{align}
  u_n\left(Z_1, \ldots, Z_n\right) &= \sum_\mathcal{P} (-)^{|\mathcal{P}| - 1} (|\mathcal{P}| - 1)! \prod_{P\in\mathcal{P}} \avg{\prod_{p\in P} Z_p}= \sum_{v=0}^{n-1} \sum_{a=0}^v (-1)^a \, a!\, \binom{n-1}{v} \stirling{v}{a} \left[ \cos(2t) \right]^{2v} \nonumber \\
  &= \sum_{v=0}^{n-1} \binom{n-1}{v} \left[ \cos(2t) \right]^{2v} \sum_{a=0}^v (-1)^a \, a!\, \stirling{v}{a} = \sum_{v=0}^{n-1} \binom{n-1}{v} \left[ \cos(2t) \right]^{2v} (-)^v \nonumber \\
  &= \sum_{v=0}^{n-1} \binom{n-1}{v} \left[ -\cos^2(2t) \right]^v = \left[ 1 - \cos^2(2t) \right]^{n-1} = \left[ \sin^2(2t) \right]^{n-1},
\end{align}
\end{widetext}
where we have used the identity $\sum_{a=0}^v (-1)^a \, a!\, \stirling{v}{a} = (-1)^v$ \citep{G10}.


\bibliographystyle{apsrev4-1}
\bibliography{multipartiteLR}

\begin{thebibliography}{32}%
\makeatletter
\providecommand \@ifxundefined [1]{%
 \@ifx{#1\undefined}
}%
\providecommand \@ifnum [1]{%
 \ifnum #1\expandafter \@firstoftwo
 \else \expandafter \@secondoftwo
 \fi
}%
\providecommand \@ifx [1]{%
 \ifx #1\expandafter \@firstoftwo
 \else \expandafter \@secondoftwo
 \fi
}%
\providecommand \natexlab [1]{#1}%
\providecommand \enquote  [1]{``#1''}%
\providecommand \bibnamefont  [1]{#1}%
\providecommand \bibfnamefont [1]{#1}%
\providecommand \citenamefont [1]{#1}%
\providecommand \href@noop [0]{\@secondoftwo}%
\providecommand \href [0]{\begingroup \@sanitize@url \@href}%
\providecommand \@href[1]{\@@startlink{#1}\@@href}%
\providecommand \@@href[1]{\endgroup#1\@@endlink}%
\providecommand \@sanitize@url [0]{\catcode `\\12\catcode `\$12\catcode
  `\&12\catcode `\#12\catcode `\^12\catcode `\_12\catcode `\%12\relax}%
\providecommand \@@startlink[1]{}%
\providecommand \@@endlink[0]{}%
\providecommand \url  [0]{\begingroup\@sanitize@url \@url }%
\providecommand \@url [1]{\endgroup\@href {#1}{\urlprefix }}%
\providecommand \urlprefix  [0]{URL }%
\providecommand \Eprint [0]{\href }%
\providecommand \doibase [0]{http://dx.doi.org/}%
\providecommand \selectlanguage [0]{\@gobble}%
\providecommand \bibinfo  [0]{\@secondoftwo}%
\providecommand \bibfield  [0]{\@secondoftwo}%
\providecommand \translation [1]{[#1]}%
\providecommand \BibitemOpen [0]{}%
\providecommand \bibitemStop [0]{}%
\providecommand \bibitemNoStop [0]{.\EOS\space}%
\providecommand \EOS [0]{\spacefactor3000\relax}%
\providecommand \BibitemShut  [1]{\csname bibitem#1\endcsname}%
\let\auto@bib@innerbib\@empty
\bibitem [{\citenamefont {Hastings}\ and\ \citenamefont {Koma}(2006)}]{HK06}%
  \BibitemOpen
  \bibfield  {author} {\bibinfo {author} {\bibfnamefont {M.~B.}\ \bibnamefont
  {Hastings}}\ and\ \bibinfo {author} {\bibfnamefont {T.}~\bibnamefont
  {Koma}},\ }\href {\doibase 10.1007/s00220-006-0030-4} {\bibfield  {journal}
  {\bibinfo  {journal} {Commun. Math. Phys.}\ }\textbf {\bibinfo {volume}
  {265}},\ \bibinfo {pages} {781} (\bibinfo {year} {2006})}\BibitemShut
  {NoStop}%
\bibitem [{\citenamefont {Foss-Feig}\ \emph {et~al.}(2015)\citenamefont
  {Foss-Feig}, \citenamefont {Gong}, \citenamefont {Clark},\ and\ \citenamefont
  {Gorshkov}}]{FGCG15}%
  \BibitemOpen
  \bibfield  {author} {\bibinfo {author} {\bibfnamefont {M.}~\bibnamefont
  {Foss-Feig}}, \bibinfo {author} {\bibfnamefont {Z.-X.}\ \bibnamefont {Gong}},
  \bibinfo {author} {\bibfnamefont {C.~W.}\ \bibnamefont {Clark}}, \ and\
  \bibinfo {author} {\bibfnamefont {A.~V.}\ \bibnamefont {Gorshkov}},\ }\href
  {\doibase 10.1103/PhysRevLett.114.157201} {\bibfield  {journal} {\bibinfo
  {journal} {Phys. Rev. Lett.}\ }\textbf {\bibinfo {volume} {114}},\ \bibinfo
  {pages} {157201} (\bibinfo {year} {2015})}\BibitemShut {NoStop}%
\bibitem [{\citenamefont {Eisert}\ \emph {et~al.}(2013)\citenamefont {Eisert},
  \citenamefont {van~den Worm}, \citenamefont {Manmana},\ and\ \citenamefont
  {Kastner}}]{EWMK13}%
  \BibitemOpen
  \bibfield  {author} {\bibinfo {author} {\bibfnamefont {J.}~\bibnamefont
  {Eisert}}, \bibinfo {author} {\bibfnamefont {M.}~\bibnamefont {van~den
  Worm}}, \bibinfo {author} {\bibfnamefont {S.~R.}\ \bibnamefont {Manmana}}, \
  and\ \bibinfo {author} {\bibfnamefont {M.}~\bibnamefont {Kastner}},\
  }\href@noop {} {\bibfield  {journal} {\bibinfo  {journal} {Phys. Rev. Lett.}\
  }\textbf {\bibinfo {volume} {111}},\ \bibinfo {pages} {260401} (\bibinfo
  {year} {2013})}\BibitemShut {NoStop}%
\bibitem [{\citenamefont {Eldredge}\ \emph {et~al.}(2016)\citenamefont
  {Eldredge}, \citenamefont {Gong}, \citenamefont {Moosavian}, \citenamefont
  {Foss-Feig},\ and\ \citenamefont {Gorshkov}}]{EGMFG16}%
  \BibitemOpen
  \bibfield  {author} {\bibinfo {author} {\bibfnamefont {Z.}~\bibnamefont
  {Eldredge}}, \bibinfo {author} {\bibfnamefont {Z.-X.}\ \bibnamefont {Gong}},
  \bibinfo {author} {\bibfnamefont {A.~H.}\ \bibnamefont {Moosavian}}, \bibinfo
  {author} {\bibfnamefont {M.}~\bibnamefont {Foss-Feig}}, \ and\ \bibinfo
  {author} {\bibfnamefont {A.~V.}\ \bibnamefont {Gorshkov}},\ }\href@noop {}
  {\bibfield  {journal} {\bibinfo  {journal} {arXiv:1612.02442}\ } (\bibinfo
  {year} {2016})}\BibitemShut {NoStop}%
\bibitem [{Note1()}]{Note1}%
  \BibitemOpen
  \bibinfo {note} {In infinite-dimensional systems, e.g.\ bosons, correlations
  may grow at arbitrarily high speed~\cite {EG09}, provided that
  non-relativistic quantum mechanics still applies.}\BibitemShut {Stop}%
\bibitem [{\citenamefont {Lieb}\ and\ \citenamefont {Robinson}(1972)}]{LR72}%
  \BibitemOpen
  \bibfield  {author} {\bibinfo {author} {\bibfnamefont {E.~H.}\ \bibnamefont
  {Lieb}}\ and\ \bibinfo {author} {\bibfnamefont {D.~W.}\ \bibnamefont
  {Robinson}},\ }\href {\doibase 10.1007/BF01645779} {\bibfield  {journal}
  {\bibinfo  {journal} {Commun. Math. Phys.}\ }\textbf {\bibinfo {volume}
  {28}},\ \bibinfo {pages} {251} (\bibinfo {year} {1972})}\BibitemShut
  {NoStop}%
\bibitem [{\citenamefont {Bravyi}\ \emph {et~al.}(2006)\citenamefont {Bravyi},
  \citenamefont {Hastings},\ and\ \citenamefont {Verstraete}}]{BHV06}%
  \BibitemOpen
  \bibfield  {author} {\bibinfo {author} {\bibfnamefont {S.}~\bibnamefont
  {Bravyi}}, \bibinfo {author} {\bibfnamefont {M.~B.}\ \bibnamefont
  {Hastings}}, \ and\ \bibinfo {author} {\bibfnamefont {F.}~\bibnamefont
  {Verstraete}},\ }\href {\doibase 10.1103/PhysRevLett.97.050401} {\bibfield
  {journal} {\bibinfo  {journal} {Phys. Rev. Lett.}\ }\textbf {\bibinfo
  {volume} {97}},\ \bibinfo {pages} {050401} (\bibinfo {year}
  {2006})}\BibitemShut {NoStop}%
\bibitem [{\citenamefont {Nachtergaele}\ \emph {et~al.}(2006)\citenamefont
  {Nachtergaele}, \citenamefont {Ogata},\ and\ \citenamefont {Sims}}]{NOS06}%
  \BibitemOpen
  \bibfield  {author} {\bibinfo {author} {\bibfnamefont {B.}~\bibnamefont
  {Nachtergaele}}, \bibinfo {author} {\bibfnamefont {Y.}~\bibnamefont {Ogata}},
  \ and\ \bibinfo {author} {\bibfnamefont {R.}~\bibnamefont {Sims}},\ }\href
  {\doibase 10.1007/s10955-006-9143-6} {\bibfield  {journal} {\bibinfo
  {journal} {Journal of Statistical Physics}\ }\textbf {\bibinfo {volume}
  {124}},\ \bibinfo {pages} {1} (\bibinfo {year} {2006})}\BibitemShut {NoStop}%
\bibitem [{\citenamefont {Nachtergaele}\ and\ \citenamefont
  {Sims}(2010)}]{NS10}%
  \BibitemOpen
  \bibfield  {author} {\bibinfo {author} {\bibfnamefont {B.}~\bibnamefont
  {Nachtergaele}}\ and\ \bibinfo {author} {\bibfnamefont {R.}~\bibnamefont
  {Sims}},\ }\href@noop {} {\bibfield  {journal} {\bibinfo  {journal}
  {arXiv:1004.2086}\ } (\bibinfo {year} {2010})}\BibitemShut {NoStop}%
\bibitem [{\citenamefont {Kliesch}\ \emph {et~al.}(2013)\citenamefont
  {Kliesch}, \citenamefont {Gogolin},\ and\ \citenamefont {Eisert}}]{KGE13}%
  \BibitemOpen
  \bibfield  {author} {\bibinfo {author} {\bibfnamefont {M.}~\bibnamefont
  {Kliesch}}, \bibinfo {author} {\bibfnamefont {C.}~\bibnamefont {Gogolin}}, \
  and\ \bibinfo {author} {\bibfnamefont {J.}~\bibnamefont {Eisert}},\
  }\href@noop {} {\bibfield  {journal} {\bibinfo  {journal} {arXiv:1306.0716}\
  } (\bibinfo {year} {2013})}\BibitemShut {NoStop}%
\bibitem [{\citenamefont {Nachtergaele}\ \emph {et~al.}(2009)\citenamefont
  {Nachtergaele}, \citenamefont {Raz}, \citenamefont {Schlein},\ and\
  \citenamefont {Sims}}]{NRSS09}%
  \BibitemOpen
  \bibfield  {author} {\bibinfo {author} {\bibfnamefont {B.}~\bibnamefont
  {Nachtergaele}}, \bibinfo {author} {\bibfnamefont {H.}~\bibnamefont {Raz}},
  \bibinfo {author} {\bibfnamefont {B.}~\bibnamefont {Schlein}}, \ and\
  \bibinfo {author} {\bibfnamefont {R.}~\bibnamefont {Sims}},\ }\href@noop {}
  {\bibfield  {journal} {\bibinfo  {journal} {Commun. Math. Phys.}\ }\textbf
  {\bibinfo {volume} {286}},\ \bibinfo {pages} {1073} (\bibinfo {year}
  {2009})}\BibitemShut {NoStop}%
\bibitem [{\citenamefont {Pr\'emont-Schwarz}\ \emph {et~al.}(2010)\citenamefont
  {Pr\'emont-Schwarz}, \citenamefont {Hnybida}, \citenamefont {Klich},\ and\
  \citenamefont {Markopoulou-Kalamara}}]{SHKM10}%
  \BibitemOpen
  \bibfield  {author} {\bibinfo {author} {\bibfnamefont {I.}~\bibnamefont
  {Pr\'emont-Schwarz}}, \bibinfo {author} {\bibfnamefont {J.}~\bibnamefont
  {Hnybida}}, \bibinfo {author} {\bibfnamefont {I.}~\bibnamefont {Klich}}, \
  and\ \bibinfo {author} {\bibfnamefont {F.}~\bibnamefont
  {Markopoulou-Kalamara}},\ }\href@noop {} {\bibfield  {journal} {\bibinfo
  {journal} {Phys. Rev. A}\ }\textbf {\bibinfo {volume} {81}},\ \bibinfo
  {pages} {040102} (\bibinfo {year} {2010})}\BibitemShut {NoStop}%
\bibitem [{\citenamefont {Pr\'emont-Schwarz}\ and\ \citenamefont
  {Hnybida}(2010)}]{SH10}%
  \BibitemOpen
  \bibfield  {author} {\bibinfo {author} {\bibfnamefont {I.}~\bibnamefont
  {Pr\'emont-Schwarz}}\ and\ \bibinfo {author} {\bibfnamefont {J.}~\bibnamefont
  {Hnybida}},\ }\href@noop {} {\bibfield  {journal} {\bibinfo  {journal} {Phys.
  Rev. A}\ }\textbf {\bibinfo {volume} {81}},\ \bibinfo {pages} {062107}
  (\bibinfo {year} {2010})}\BibitemShut {NoStop}%
\bibitem [{\citenamefont {Cheneau}\ \emph {et~al.}(2012)\citenamefont
  {Cheneau}, \citenamefont {Barmettler}, \citenamefont {Poletti}, \citenamefont
  {Endres}, \citenamefont {Schau{\ss}}, \citenamefont {Fukuhara}, \citenamefont
  {Gross}, \citenamefont {Bloch}, \citenamefont {Kollath},\ and\ \citenamefont
  {Kuhr}}]{CBPESFGBKK}%
  \BibitemOpen
  \bibfield  {author} {\bibinfo {author} {\bibfnamefont {M.}~\bibnamefont
  {Cheneau}}, \bibinfo {author} {\bibfnamefont {P.}~\bibnamefont {Barmettler}},
  \bibinfo {author} {\bibfnamefont {D.}~\bibnamefont {Poletti}}, \bibinfo
  {author} {\bibfnamefont {M.}~\bibnamefont {Endres}}, \bibinfo {author}
  {\bibfnamefont {P.}~\bibnamefont {Schau{\ss}}}, \bibinfo {author}
  {\bibfnamefont {T.}~\bibnamefont {Fukuhara}}, \bibinfo {author}
  {\bibfnamefont {C.}~\bibnamefont {Gross}}, \bibinfo {author} {\bibfnamefont
  {I.}~\bibnamefont {Bloch}}, \bibinfo {author} {\bibfnamefont
  {C.}~\bibnamefont {Kollath}}, \ and\ \bibinfo {author} {\bibfnamefont
  {S.}~\bibnamefont {Kuhr}},\ }\href@noop {} {\bibfield  {journal} {\bibinfo
  {journal} {Nature}\ }\textbf {\bibinfo {volume} {481}},\ \bibinfo {pages}
  {484} (\bibinfo {year} {2012})}\BibitemShut {NoStop}%
\bibitem [{\citenamefont {Richerme}\ \emph {et~al.}(2014)\citenamefont
  {Richerme}, \citenamefont {Gong}, \citenamefont {Lee}, \citenamefont {Senko},
  \citenamefont {Smith}, \citenamefont {Foss-Feig}, \citenamefont {Michalakis},
  \citenamefont {Gorshkov},\ and\ \citenamefont {Monroe}}]{RGLSSFMGM14}%
  \BibitemOpen
  \bibfield  {author} {\bibinfo {author} {\bibfnamefont {P.}~\bibnamefont
  {Richerme}}, \bibinfo {author} {\bibfnamefont {Z.-X.}\ \bibnamefont {Gong}},
  \bibinfo {author} {\bibfnamefont {A.}~\bibnamefont {Lee}}, \bibinfo {author}
  {\bibfnamefont {C.}~\bibnamefont {Senko}}, \bibinfo {author} {\bibfnamefont
  {J.}~\bibnamefont {Smith}}, \bibinfo {author} {\bibfnamefont
  {M.}~\bibnamefont {Foss-Feig}}, \bibinfo {author} {\bibfnamefont
  {S.}~\bibnamefont {Michalakis}}, \bibinfo {author} {\bibfnamefont {A.~V.}\
  \bibnamefont {Gorshkov}}, \ and\ \bibinfo {author} {\bibfnamefont
  {C.}~\bibnamefont {Monroe}},\ }\href {\doibase doi:10.1038/nature13450}
  {\bibfield  {journal} {\bibinfo  {journal} {Nature}\ }\textbf {\bibinfo
  {volume} {511}},\ \bibinfo {pages} {198} (\bibinfo {year}
  {2014})}\BibitemShut {NoStop}%
\bibitem [{\citenamefont {Zhou}\ \emph {et~al.}(2006)\citenamefont {Zhou},
  \citenamefont {Zeng}, \citenamefont {Xu},\ and\ \citenamefont {You}}]{ZZX06}%
  \BibitemOpen
  \bibfield  {author} {\bibinfo {author} {\bibfnamefont {D.~L.}\ \bibnamefont
  {Zhou}}, \bibinfo {author} {\bibfnamefont {B.}~\bibnamefont {Zeng}}, \bibinfo
  {author} {\bibfnamefont {Z.}~\bibnamefont {Xu}}, \ and\ \bibinfo {author}
  {\bibfnamefont {L.}~\bibnamefont {You}},\ }\href {\doibase
  10.1103/PhysRevA.74.052110} {\bibfield  {journal} {\bibinfo  {journal} {Phys.
  Rev. A}\ }\textbf {\bibinfo {volume} {74}},\ \bibinfo {pages} {052110}
  (\bibinfo {year} {2006})}\BibitemShut {NoStop}%
\bibitem [{\citenamefont {Schweigler}\ \emph {et~al.}(2017)\citenamefont
  {Schweigler}, \citenamefont {Kasper}, \citenamefont {Erne}, \citenamefont
  {Mazets}, \citenamefont {Rauer}, \citenamefont {Cataldini}, \citenamefont
  {Langen}, \citenamefont {Gasenzer}, \citenamefont {Berges},\ and\
  \citenamefont {Schmiedmayer}}]{SKEMRCLGBS17}%
  \BibitemOpen
  \bibfield  {author} {\bibinfo {author} {\bibfnamefont {T.}~\bibnamefont
  {Schweigler}}, \bibinfo {author} {\bibfnamefont {V.}~\bibnamefont {Kasper}},
  \bibinfo {author} {\bibfnamefont {S.}~\bibnamefont {Erne}}, \bibinfo {author}
  {\bibfnamefont {I.}~\bibnamefont {Mazets}}, \bibinfo {author} {\bibfnamefont
  {B.}~\bibnamefont {Rauer}}, \bibinfo {author} {\bibfnamefont
  {F.}~\bibnamefont {Cataldini}}, \bibinfo {author} {\bibfnamefont
  {T.}~\bibnamefont {Langen}}, \bibinfo {author} {\bibfnamefont
  {T.}~\bibnamefont {Gasenzer}}, \bibinfo {author} {\bibfnamefont
  {J.}~\bibnamefont {Berges}}, \ and\ \bibinfo {author} {\bibfnamefont
  {J.}~\bibnamefont {Schmiedmayer}},\ }\href
  {http://dx.doi.org/10.1038/nature22310} {\bibfield  {journal} {\bibinfo
  {journal} {Nature}\ }\textbf {\bibinfo {volume} {545}},\ \bibinfo {pages}
  {323} (\bibinfo {year} {2017})}\BibitemShut {NoStop}%
\bibitem [{\citenamefont {Horodecki}\ \emph {et~al.}(2009)\citenamefont
  {Horodecki}, \citenamefont {Horodecki}, \citenamefont {Horodecki},\ and\
  \citenamefont {Horodecki}}]{HHHH09}%
  \BibitemOpen
  \bibfield  {author} {\bibinfo {author} {\bibfnamefont {R.}~\bibnamefont
  {Horodecki}}, \bibinfo {author} {\bibfnamefont {P.}~\bibnamefont
  {Horodecki}}, \bibinfo {author} {\bibfnamefont {M.}~\bibnamefont
  {Horodecki}}, \ and\ \bibinfo {author} {\bibfnamefont {K.}~\bibnamefont
  {Horodecki}},\ }\href {\doibase 10.1103/RevModPhys.81.865} {\bibfield
  {journal} {\bibinfo  {journal} {Rev. Mod. Phys.}\ }\textbf {\bibinfo {volume}
  {81}},\ \bibinfo {pages} {865} (\bibinfo {year} {2009})}\BibitemShut
  {NoStop}%
\bibitem [{\citenamefont {Ursell}(1927)}]{U27}%
  \BibitemOpen
  \bibfield  {author} {\bibinfo {author} {\bibfnamefont {H.~D.}\ \bibnamefont
  {Ursell}},\ }\href {\doibase 10.1017/S0305004100011191} {\bibfield  {journal}
  {\bibinfo  {journal} {Math. Proc. Cambridge Philos. Soc.}\ }\textbf {\bibinfo
  {volume} {23}},\ \bibinfo {pages} {685} (\bibinfo {year} {1927})}\BibitemShut
  {NoStop}%
\bibitem [{\citenamefont {Sylvester}(1975)}]{S75}%
  \BibitemOpen
  \bibfield  {author} {\bibinfo {author} {\bibfnamefont {G.~S.}\ \bibnamefont
  {Sylvester}},\ }\href {\doibase 10.1007/BF01608973} {\bibfield  {journal}
  {\bibinfo  {journal} {Commun. Math. Phys.}\ }\textbf {\bibinfo {volume}
  {42}},\ \bibinfo {pages} {209} (\bibinfo {year} {1975})}\BibitemShut
  {NoStop}%
\bibitem [{\citenamefont {Hauser}\ \emph {et~al.}(1996)\citenamefont {Hauser},
  \citenamefont {Cassing}, \citenamefont {Peter},\ and\ \citenamefont
  {Thoma}}]{H94}%
  \BibitemOpen
  \bibfield  {author} {\bibinfo {author} {\bibfnamefont {J.~M.}\ \bibnamefont
  {Hauser}}, \bibinfo {author} {\bibfnamefont {W.}~\bibnamefont {Cassing}},
  \bibinfo {author} {\bibfnamefont {A.}~\bibnamefont {Peter}}, \ and\ \bibinfo
  {author} {\bibfnamefont {M.~H.}\ \bibnamefont {Thoma}},\ }\href {\doibase
  10.1007/BF01292336} {\bibfield  {journal} {\bibinfo  {journal} {Z. Phys.}\
  }\textbf {\bibinfo {volume} {A353}},\ \bibinfo {pages} {301} (\bibinfo {year}
  {1996})},\ \Eprint {http://arxiv.org/abs/hep-ph/9408355}
  {arXiv:hep-ph/9408355 [hep-ph]} \BibitemShut {NoStop}%
\bibitem [{\citenamefont {Kopietz}\ \emph {et~al.}(2010)\citenamefont
  {Kopietz}, \citenamefont {Bartosch},\ and\ \citenamefont {Sch{\"u}tz}}]{K10}%
  \BibitemOpen
  \bibfield  {author} {\bibinfo {author} {\bibfnamefont {P.}~\bibnamefont
  {Kopietz}}, \bibinfo {author} {\bibfnamefont {L.}~\bibnamefont {Bartosch}}, \
  and\ \bibinfo {author} {\bibfnamefont {F.}~\bibnamefont {Sch{\"u}tz}},\
  }\enquote {\bibinfo {title} {{Mean-Field Theory and the Gaussian
  Approximation}},}\ in\ \href {\doibase 10.1007/978-3-642-05094-7_2} {\emph
  {\bibinfo {booktitle} {Introduction to the Functional Renormalization
  Group}}}\ (\bibinfo  {publisher} {Springer Berlin Heidelberg},\ \bibinfo
  {address} {Berlin, Heidelberg},\ \bibinfo {year} {2010})\ pp.\ \bibinfo
  {pages} {23--52}\BibitemShut {NoStop}%
\bibitem [{\citenamefont {Walter}\ \emph {et~al.}(2016)\citenamefont {Walter},
  \citenamefont {Gross},\ and\ \citenamefont {Eisert}}]{WGE16}%
  \BibitemOpen
  \bibfield  {author} {\bibinfo {author} {\bibfnamefont {M.}~\bibnamefont
  {Walter}}, \bibinfo {author} {\bibfnamefont {D.}~\bibnamefont {Gross}}, \
  and\ \bibinfo {author} {\bibfnamefont {J.}~\bibnamefont {Eisert}},\ }\href
  {https://arxiv.org/abs/1612.02437} {\bibfield  {journal} {\bibinfo  {journal}
  {arXiv:1612.02437}\ } (\bibinfo {year} {2016})}\BibitemShut {NoStop}%
\bibitem [{\citenamefont {Pappalardi}\ \emph {et~al.}(2017)\citenamefont
  {Pappalardi}, \citenamefont {Russomanno}, \citenamefont {Silva},\ and\
  \citenamefont {Fazio}}]{PRSF17}%
  \BibitemOpen
  \bibfield  {author} {\bibinfo {author} {\bibfnamefont {S.}~\bibnamefont
  {Pappalardi}}, \bibinfo {author} {\bibfnamefont {A.}~\bibnamefont
  {Russomanno}}, \bibinfo {author} {\bibfnamefont {A.}~\bibnamefont {Silva}}, \
  and\ \bibinfo {author} {\bibfnamefont {R.}~\bibnamefont {Fazio}},\ }\href
  {https://arxiv.org/abs/1612.02437} {\bibfield  {journal} {\bibinfo  {journal}
  {arXiv:1701.05883}\ } (\bibinfo {year} {2017})}\BibitemShut {NoStop}%
\bibitem [{Note2()}]{Note2}%
  \BibitemOpen
  \bibinfo {note} {See Eq.\ (6) of Ref.~\cite {RGLSSFMGM14}}\BibitemShut
  {NoStop}%
\bibitem [{\citenamefont {Hein}\ \emph {et~al.}(2004)\citenamefont {Hein},
  \citenamefont {Eisert},\ and\ \citenamefont {Briegel}}]{Hein04}%
  \BibitemOpen
  \bibfield  {author} {\bibinfo {author} {\bibfnamefont {M.}~\bibnamefont
  {Hein}}, \bibinfo {author} {\bibfnamefont {J.}~\bibnamefont {Eisert}}, \ and\
  \bibinfo {author} {\bibfnamefont {H.~J.}\ \bibnamefont {Briegel}},\ }\href
  {\doibase 10.1103/PhysRevA.69.062311} {\bibfield  {journal} {\bibinfo
  {journal} {Phys. Rev. A}\ }\textbf {\bibinfo {volume} {69}},\ \bibinfo
  {pages} {062311} (\bibinfo {year} {2004})}\BibitemShut {NoStop}%
\bibitem [{\citenamefont {Raussendorf}\ \emph {et~al.}(2003)\citenamefont
  {Raussendorf}, \citenamefont {Browne},\ and\ \citenamefont
  {Briegel}}]{RBB03}%
  \BibitemOpen
  \bibfield  {author} {\bibinfo {author} {\bibfnamefont {R.}~\bibnamefont
  {Raussendorf}}, \bibinfo {author} {\bibfnamefont {D.~E.}\ \bibnamefont
  {Browne}}, \ and\ \bibinfo {author} {\bibfnamefont {H.~J.}\ \bibnamefont
  {Briegel}},\ }\href {\doibase 10.1103/PhysRevA.68.022312} {\bibfield
  {journal} {\bibinfo  {journal} {Phys. Rev. A}\ }\textbf {\bibinfo {volume}
  {68}},\ \bibinfo {pages} {022312} (\bibinfo {year} {2003})}\BibitemShut
  {NoStop}%
\bibitem [{\citenamefont {Nielsen}(2006)}]{N06}%
  \BibitemOpen
  \bibfield  {author} {\bibinfo {author} {\bibfnamefont {M.~A.}\ \bibnamefont
  {Nielsen}},\ }\href {\doibase
  http://dx.doi.org/10.1016/S0034-4877(06)80014-5} {\bibfield  {journal}
  {\bibinfo  {journal} {Rep. Math. Phys.}\ }\textbf {\bibinfo {volume} {57}},\
  \bibinfo {pages} {147 } (\bibinfo {year} {2006})}\BibitemShut {NoStop}%
\bibitem [{\citenamefont {Kastner}(2015)}]{K15}%
  \BibitemOpen
  \bibfield  {author} {\bibinfo {author} {\bibfnamefont {M.}~\bibnamefont
  {Kastner}},\ }\href {http://stacks.iop.org/1367-2630/17/i=12/a=123024}
  {\bibfield  {journal} {\bibinfo  {journal} {New J. Phys.}\ }\textbf {\bibinfo
  {volume} {17}},\ \bibinfo {pages} {123024} (\bibinfo {year}
  {2015})}\BibitemShut {NoStop}%
\bibitem [{\citenamefont {G\"arttner}\ \emph {et~al.}(2016)\citenamefont
  {G\"arttner}, \citenamefont {Bohnet}, \citenamefont {Safavi-Naini},
  \citenamefont {Wall}, \citenamefont {Bollinger},\ and\ \citenamefont
  {Rey}}]{GBSWBR16}%
  \BibitemOpen
  \bibfield  {author} {\bibinfo {author} {\bibfnamefont {M.}~\bibnamefont
  {G\"arttner}}, \bibinfo {author} {\bibfnamefont {J.~G.}\ \bibnamefont
  {Bohnet}}, \bibinfo {author} {\bibfnamefont {A.}~\bibnamefont
  {Safavi-Naini}}, \bibinfo {author} {\bibfnamefont {M.~L.}\ \bibnamefont
  {Wall}}, \bibinfo {author} {\bibfnamefont {J.~J.}\ \bibnamefont {Bollinger}},
  \ and\ \bibinfo {author} {\bibfnamefont {A.~M.}\ \bibnamefont {Rey}},\ }\href
  {https://arxiv.org/abs/1608.08938} {\bibfield  {journal} {\bibinfo  {journal}
  {arXiv:1608.08938}\ } (\bibinfo {year} {2016})}\BibitemShut {NoStop}%
\bibitem [{\citenamefont {J.}\ and\ \citenamefont {Gould}()}]{G10}%
  \BibitemOpen
  \bibfield  {author} {\bibinfo {author} {\bibfnamefont {Q.}~\bibnamefont
  {J.}}\ and\ \bibinfo {author} {\bibfnamefont {H.~W.}\ \bibnamefont {Gould}},\
  }\href@noop {} {\emph {\bibinfo {title} {Combinatorial Identities for
  Stirling Numbers The Unpublished Notes of H W Gould}}}\ (\bibinfo
  {publisher} {World Scientific Publishing Co., Singapore})\BibitemShut
  {NoStop}%
\bibitem [{\citenamefont {Eisert}\ and\ \citenamefont {Gross}(2009)}]{EG09}%
  \BibitemOpen
  \bibfield  {author} {\bibinfo {author} {\bibfnamefont {J.}~\bibnamefont
  {Eisert}}\ and\ \bibinfo {author} {\bibfnamefont {D.}~\bibnamefont {Gross}},\
  }\href@noop {} {\bibfield  {journal} {\bibinfo  {journal} {Phys. Rev. Lett.}\
  }\textbf {\bibinfo {volume} {102}},\ \bibinfo {pages} {240501} (\bibinfo
  {year} {2009})}\BibitemShut {NoStop}%
\end{thebibliography}%

\end{document}